\documentclass[submission,copyright,creativecommons]{eptcs}
\providecommand{\event}{AFL 2023} %

\usepackage{iftex}

\ifpdf
  \usepackage{underscore}         %
  \usepackage[T1]{fontenc}        %
\else
  \usepackage{breakurl}           %
\fi

\usepackage{amssymb}
\usepackage{amsmath}
\usepackage[onelanguage,linesnumbered,vlined,algoruled,noresetcount,procnumbered]{algorithm2e}
\SetVlineSkip{0.0ex}\DontPrintSemicolon
\SetProcNameSty{texttt}
\SetKw{Halt}{halt}
\SetKw{Break}{break}
\SetKw{KwNot}{not}
\SetKw{KwAnd}{and}
\SetKw{KwOr}{or}
\SetKw{KwNil}{nil}
\SetKw{KwFrom}{from}
\SetKwFunction{readSymbol}{read}
\SetKwFunction{accept}{Accept}
\SetKwFunction{reject}{Reject}
\SetKwData{beforeMarking}{beforeMarking}
\SetKwData{afterMarking}{afterMarking}
\SetKwData{simulation}{backwardSimulation}

\SetKwData{markedSymbol}{\ensuremath{\sigma}}
\SetKwData{markingState}{\ensuremath{s}}
\SetKwData{simulatedState}{\ensuremath{p}}
\SetKwData{simulatedSymbol}{\ensuremath{X}}
\SetKwData{simulatedDirection}{\ensuremath{d}}
\SetKwData{mode}{mode}
\SetKwData{scannedSymbol}{scannedSymbol}

\SetKwData{true}{true}
\SetKwData{false}{false}
\usepackage[capitalize,noabbrev,sort&compress,nameinlink]{cleveref}
\crefformat{footnote}{#2\footnotemark[#1]#3}
\crefname{algorithm}{Procedure}{Procedures}
\usepackage{ifthen}
\usepackage{xspace}
\usepackage{xparse,xifthen,xspace,xstring}

\usepackage{pifont,xcolor}
\definecolor{pinegreen}{cmyk}{0.92,0,0.59,0.25} %

\usepackage{accents}
\newcommand\marked[1]{{\ensuremath{\accentset\bullet{#1}}}}

\usepackage{xspace}

\newcommand{\set}[1]{\mbox{$\{#1\}$}}

\newcommand\poly{{\rm poly}}

\ProvideDocumentCommand\varconstr{ s O{} m O{} }{%
	\ensuremath{%
		\IfBooleanTF{#1}{%
			\IfInteger{#2}%
			{{{#4 #3}^{\multido{}{#2}{\prime}}}}%
			{{{#4 #3}\ifthenelse{\isempty{#2}}{'}{^{#2}}}}%
		}%
		{%
			{#4 #3}\ifthenelse{\isempty{#2}}{}{_{#2}}%
		}%
	}%
	\xspace%
}
\ProvideDocumentCommand\mA{ s O{A} O{} O{\cal} }{%
	\IfBooleanTF{#1}%
	{\varconstr*[#3]{#2}[#4]}%
	{\varconstr[#3]{#2}[#4]}%
}

\newcommand\s{s\xspace}
\newcommand\ow{\textsc1\xspace}
\newcommand\tw{\textsc2\xspace}

\newcommand\om{\textsc{om}\xspace}
\newcommand\am{\textsc{am}\xspace}

\newcommand\dfa{\textsc{dfa}\xspace}
\newcommand\nfa{\textsc{nfa}\xspace}
\newcommand\twdfa{{\tw}{\dfa}\xspace}
\newcommand\twnfa{{\tw}{\nfa}\xspace}
\newcommand\owdfa{{\ow}{\dfa}\xspace}
\newcommand\ownfa{{\ow}{\nfa}\xspace}

\newcommand\twdfas{{\twdfa}\s}
\newcommand\twnfas{{\twnfa}\s}
\newcommand\owdfas{{\owdfa}\s}
\newcommand\ownfas{{\ownfa}\s}
\newcommand\domlas{{\domla}\s}
\newcommand\damlas{{\damla}\s}
\newcommand\amlas{{\amla}\s}
\newcommand\omlas{{\omla}\s}

\newcommand*\la[1][1]{\ensuremath{#1}\textsc{\ifx&#1&\else-\fi la}\xspace}
\newcommand*\dla[1][1]{\textsc{d-}\la[#1]}%

\newcommand*\omla[1][1]{\om-\ensuremath{#1}\textsc{\ifx&#1&\else-\fi la}\xspace}
\newcommand*\domla[1][1]{\textsc{d-}\omla[#1]}%

\newcommand*\amla[1][1]{\am-\ensuremath{#1}\textsc{\ifx&#1&\else-\fi la}\xspace}
\newcommand*\damla[1][1]{\textsc{d-}\amla[#1]}%

\newcommand*\las[1][1]{{\la[#1]}\s}   %
\newcommand*\dlas[1][1]{{\dla[#1]}\s} %
\newcommand*\lend{{\ensuremath{\mathord{\vartriangleright}}}\xspace}
\newcommand*\rend{{\ensuremath{\mathord{\vartriangleleft}}}\xspace}

\newcommand\bigoof[1]{\ensuremath{O\left(#1\right)}}

\usepackage{amsthm}
\makeatletter
\newtheorem*{rep@theorem}{\rep@title}
\newcommand{\newreptheorem}[2]{%
\newenvironment{rep#1}[1]{%
 \def\rep@title{#2 \ref{##1}}%
 \begin{rep@theorem}}%
 {\end{rep@theorem}}}
\makeatother

\newtheorem{theorem}{Theorem}

\newtheorem{definition}{Definition}
\newtheorem{ex}{Example}

\newenvironment{example-cont}[1]{\bigskip\noindent\textbf{Example~\ref{#1}.~(cont.)\hspace{\labelsep}}}{\bigskip\noindent}

\newreptheorem{theorem}{Theorem}

\title{Once-Marking and Always-Marking $1$-Limited Automata}
\author{Giovanni Pighizzini
\institute{Dipartimento di Informatica\\
Universit\`{a} degli Studi di Milano, Italy}
\email{pighizzini@di.unimi.it}
\and
Luca Prigioniero
\institute{Department of Computer Science\\
Loughborough University, UK}
\email{l.prigioniero@lboro.ac.uk}
}
\def\titlerunning{Once-Marking and Always-Marking $1$-Limited Automata}
\def\authorrunning{G.~Pighizzini \& L.~Prigioniero}

\begin{document}
\maketitle

\begin{abstract}
\noindent
Single-tape nondeterministic Turing machines that are allowed to replace the symbol in each tape
cell only when it is scanned for the first time
are also known as $1$-limited automata.
These devices characterize,
exactly as finite automata,
the class of regular languages.
However, they can be extremely more succinct.
Indeed,
in the worst case the size gap from $1$-limited automata to one-way deterministic finite automata is double exponential.

Here we introduce two restricted versions of $1$-limited automata, 
\emph{once-marking $1$-limited automata} and \emph{always-marking $1$-limited automata}, 
and study their descriptional complexity. 
We prove that once-marking $1$-limited automata still exhibit a double exponential size gap
to one-way deterministic finite automata. However, their deterministic
restriction is polynomially related in size to two-way deterministic finite automata, in contrast to deterministic
$1$-limited automata, whose equivalent two-way deterministic finite automata in the worst case
are exponentially larger.
For always-marking $1$-limited automata, we prove that the size gap to one-way deterministic finite automata
is only a single exponential. The gap remains exponential even in the case the given machine is deterministic.

We obtain other size relationships between different variants of these machines and finite automata and
we present some problems that deserve investigation.
\end{abstract}

\section{Introduction}
\label{sec:intro}

In 1967, with the aim of generalizing the concept of determinism for context-free languages,
Hibbard introduced \emph{limited automata}, a restricted version of Turing machines~\cite{Hi67}.
More precisely, for each fixed integer~$d\geq 0$, a \emph{$d$-limited automaton} is a single-tape nondeterministic Turing machine that
is allowed to replace the content of each tape cell only in the first~$d$ visits.

Hibbard proved that, for each~$d\geq 2$, $d$-limited automata characterize the class of context-free languages.
For~$d=0$ these devices cannot modify the input tape, hence they are two-way finite automata, so characterizing regular languages.
Furthermore, also $1$-limited automata are no more powerful than finite automata. The proof of this fact can be
found in~\cite[Thm.~12.1]{WW86}.

The investigation of these models has been reconsidered in the last decade, mainly from a descriptional point of view.
Starting with~\cite{PP14,PP15}, several works investigating properties of limited automata
and their relationships with other computational models appeared in the literature (for a recent survey see~\cite{Pig19}).

In this paper we focus on $1$-limited automata.
We already mentioned that these devices are no more powerful than finite automata, namely they recognize the class
of regular languages. However, they can be dramatically more succinct than finite automata. In fact, a double exponential
size gap from 1-limited automata to one-way deterministic finite automata has been proved~\cite{PP14}.
In other words, every $n$-state 1-limited automaton can be simulated by a one-way deterministic automaton with a
number of states which is double exponential in~$n$. Furthermore, in the worst case, this cost cannot be reduced.

As pointed out in~\cite{PP14}, this double exponential gap is related to a double role of the nondeterminism in 
1-limited automata. When the head of a 1-limited automaton reaches for the first time a tape cell, it replaces the symbol
in it according to a nondeterministic choice.
Furthermore, the set of nondeterministic choices allowed during the next visits to the same cell depends 
on the symbol written in the first visit and that cannot be further changed, 
namely it depends on the nondeterministic choice made during the first visit.

With the aim of better understanding this phenomenon, we started to investigate some restrictions of 1-limited automata.
On the one hand, we are interested in finding restrictions that reduce this double exponential gap to a single exponential.
We already know that this happens for \emph{deterministic} $1$-limited automata~\cite{PP14}.
So the problem is finding some restrictions that, still allowing nondeterministic transitions, avoid the double exponential gap.
On the other hand, we are also interested in finding some very restricted forms of $1$-limited automata for which a double
exponential size gap in the conversion to one-way deterministic automata remains necessary in the worst case.

A first attempt could be requiring deterministic rewritings,
according to the current configuration of the machine,
every time cells are visited for the first time,
still keeping nondeterministic the choice of the next 
state and head movement.
Another attempt could be to allow nondeterministic choices
for the symbol to rewrite,
but not for the next state and the head movement. 
In both cases the double exponential gap to one-way deterministic finite automata remains possible.
Indeed, in both cases,
different computation paths can replace the same input prefix on the tape with different strings,
as in the original model.
Actually, we noticed that the double exponential gap can be achieved already for $1$-limited automata that,
in each computation,
have the possibility to mark just one tape cell leaving the rest of the tape unchanged.
This inspired us to investigate machines with such a restriction,
which we call \emph{once-marking $1$-limited automata}.
We show that
the double exponential size gap to one-way deterministic finite automata remains possible
even for once-marking $1$-limited automata that are \emph{sweeping}
(namely,
change the head direction only at the left or right end of the tape)
and that are allowed to use nondeterminism only in the first visit to tape cells.
Comparing the size of once-marking $1$-limited automata with other kinds of finite automata, 
we prove an exponential gap to two-way nondeterministic automata.
The situation changes significantly when nondeterministic transitions are not possible.
Indeed,
we prove that
every deterministic once-marking $1$-limited automaton can be converted into an equivalent two-way deterministic finite automaton
with only a polynomial size increasing.
The costs we obtain concerning once-marking $1$-limited automata are summarized in \cref{fig:diagramOM}.

As mentioned above,
the double exponential gap from $1$-limited automata to one-way deterministic finite automata is related to the fact that
different computation paths can replace the same input prefix on the tape with different strings.
This suggested the idea of considering a different restriction,
which prevents this possibility,
by requiring the replacement of each input symbol~$a$ with a symbol that depends only on~$a$. 
To this aim, here we introduce \emph{always-marking $1$-limited automata}, that in the
first visit replace each symbol with a marked version of it. We show that in this case the gap from these devices,
in the nondeterministic version, to one-way deterministic finite automata reduces to a single exponential.
The same gap holds when converting always-marking $1$-limited automata into one-way nondeterministic finite automata,
but even when converting \emph{deterministic} always-marking $1$-limited automata into \emph{two-way nondeterministic} finite automata.
The bounds we obtain concerning always-marking $1$-limited automata are summarized in \cref{fig:diagramAM}.

\medskip

The paper is organized as follows.
After presenting in \cref{sec:prel} the preliminary notions used in the paper and,
in particular,
the definition of $1$-limited automata with the fundamental results on their descriptional complexity,
in \cref{sec:variant} we introduce once-marking and always-marking $1$-limited automata,
together with some witness languages that will be useful to obtain our results.
\cref{sec:OM,sec:AM} are devoted to the investigation of the descriptional complexity of these models.
We conclude the paper presenting some final remarks and possible lines for future investigations.

\section{Preliminaries}
\label{sec:prel}
In this section we recall some basic definitions useful in the paper.
Given a set~$S$,
$\#{S}$~denotes its cardinality and~$2^S$ the family of all its subsets.
Given an alphabet~$\Sigma$, a string~$w\in\Sigma^*$, and a symbol~$a\in\Sigma$,
$|w|$ denotes the length of~$w$,
$\Sigma^k$ the set of all strings on~$\Sigma$ of length~$k$,
$\marked a$ the \emph{marked versions} of~$a$,
and~$\marked\Sigma = \set{\marked a \mid a \in \Sigma}$ the set of the marked versions of the symbols in~$\Sigma$.

We assume the reader familiar with notions from formal languages and automata
theory, in particular with the fundamental variants of finite automata
(\owdfas, \ownfas, \twdfas, \twnfas, for short, where \ow/\tw mean
\emph{one-way}/\emph{two-way} and \textsc{d}/\textsc{n} mean
\emph{deterministic}/\emph{nondeterministic}, respectively).
For any unfamiliar terminology see, e.g.,~\cite{HU79}.

A \emph{$1$-limited automaton} (\la, for short)
is a tuple $\mA=(Q,\Sigma,\Gamma,\delta,q_I,F)$,
where~$Q$ is a finite \emph{set of states},
$\Sigma$ is a finite \emph{input alphabet},
$\Gamma$ is a finite \emph{work alphabet} such that~$\Sigma \cup
\{\lend,\rend\} \subseteq \Gamma$,
$\lend,\rend \notin \Sigma$
are two special symbols,
called the \emph{left} and the \emph{right end-markers},
and~$\delta:Q\times\Gamma\rightarrow 2^{Q\times(\Gamma\setminus\{\lend,\rend\})\times\{-1,+1\}}$ is the \emph{transition function}.
At the beginning of the computation,
the input word~$w\in\Sigma^*$ is stored onto the tape surrounded by the two end-markers,
the left end-marker being in position zero
and
the right end-marker being in position~$|w|+1$.
The head of the automaton is on cell~$1$ and the state of the finite control is the \emph{initial state}~$q_I$.

In one move,
according to~$\delta$ and the current state,
\mA\ reads a symbol from the tape,
changes its state,
replaces the symbol just read from the tape by a new symbol,
and moves its head to one position forward or backward.
Furthermore, the head cannot pass the end-markers,
except at the end of computation,
to accept the input, as explained below.
Replacing symbols is allowed to modify the content of each cell only
during the first visit,
with the exception of the cells containing the end-markers,
which are never modified.
Hence, after the first visit, a tape cell is ``frozen''.%
\footnote{%
	More technical details can be found in~\cite{PP14}.
	However,
	a syntactical restriction forcing \las to replace in the first visit to each tape cell the input symbol in it with another symbol from an alphabet~$\Gamma_1$ disjoint
	from~$\Sigma$,
	was given.
	Here we drop this restriction,
	in order to be able to see once-marking \las as a restriction of \las.
	It is always possible to transform a \la into an equivalent \la satisfying such a syntactical restriction,
	just extending~$\Gamma$ with a marked copy of~$\Sigma$ and suitably modifying the transition function.
}

The automaton \mA accepts an input~$w$ if and only if there is a computation path that
starts from the initial state~$q_I$ with the input tape containing~$w$
surrounded by the two end-markers and the head on the first input cell,
and that ends in a \emph{final state}~$q\in F$ after passing the right end-marker.
The device \mA~is said to be \emph{deterministic} (\dla, for short) whenever~$\#{\delta(q,\sigma)}\le 1$, for any $q\in Q$ and $\sigma\in\Gamma$.

Two-way finite automata are limited automata in which no rewritings are possible.
On the other hand,
one-way finite automata can scan the input in a one-way fashion only.
A finite automaton is, as usual, a tuple~$(Q,\Sigma,\delta,q_I,F)$,
where,
analogously to \las,
~$Q$ is the finite set of states,
$\Sigma$ is the finite input alphabet,
$\delta$ is the transition function,
$q_I$ is the initial state,
and~$F$ is the set of final states.
We point out that for two-way finite automata we assume the same accepting conditions as for \las. 

Two-way machines 
in which the direction of the head can change only at the end-markers
are said to be \emph{sweeping}~\cite{Sip80b}.

\medskip

In this paper we are interested to compare the size of machines.
The \emph{size} of a model
is given by the total number of symbols used to write down its description.
Therefore,
the size of \las is bounded by a polynomial
in the number of states and of work symbols,
while,
in the case of finite automata,
since no writings are allowed,
the size is linear in the number of instructions and states,
which is bounded by a polynomial in the number of states
and in the number of input symbols.

The size costs of the simulations from \las to finite automata have been studied in~\cite{PP14} and are 
summarized in \cref{fig:diagram}.
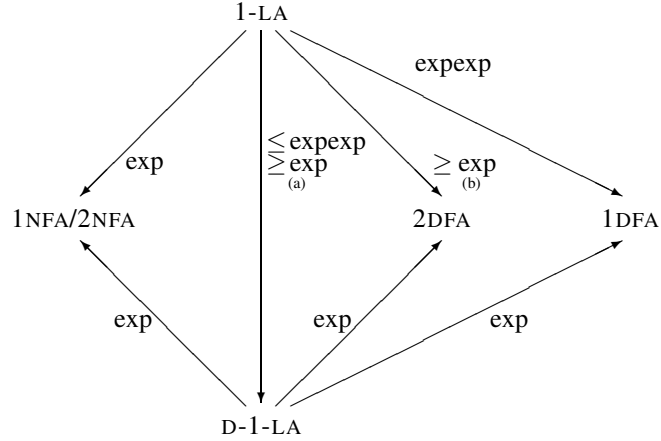
\begin{figure}[t]
	\setlength{\unitlength}{0.5cm}
\centering
	\begin{picture}(17,12)(0,0)
	
	\put(6.5,0.5){\makebox(0,0){\small\sc \dla}}
	\put(1.5,6){\makebox(0,0){\small\ownfa/\twnfa}}
	\put(11.3,6){\makebox(0,0){\small\twdfa}}
	\put(16.3,6){\makebox(0,0){\small\owdfa}}
	\put(6.5,11.5){\makebox(0,0){\small\la}}

	\put(6,11){\vector(-1,-1){4.4}} 
	\put(3.3,7.4){\makebox(0,0){\small$\exp$}}
	\put(6.4,11){\vector(0,-1){9.9}} 
	\put(7.85,8){\makebox(0,0){\small$\leq\!\text{expexp}$}}
	\put(7.35,7.4){\makebox(0,0){\small$\geq\!\exp$}}
	\put(7.35,6.9){\makebox(0,0){\tiny(a)}}

	\put(6.8,11){\vector(1,-1){4.4}} 
	\put(11.8,7.4){\makebox(0,0){\small$\geq\exp$}}
	\put(12,6.9){\makebox(0,0){\tiny(b)}}

	\put(7.2,11){\vector(2,-1){8.8}} 
	\put(11.5,10){\makebox(0,0){\small$\text{expexp}$}}
	\put(6,1){\vector(-1,1){4.4}} 
	\put(3,3.2){\makebox(0,0){\small$\exp$}}
	\put(6.8,1){\vector(1,1){4.4}} 
	\put(8.3,3.2){\makebox(0,0){\small$\exp$}}
	\put(7.2,1){\vector(2,1){8.8}} 
	\put(13,3.2){\makebox(0,0){\small$\exp$}}

	\end{picture}
\caption{Size costs of conversions of \las and \dlas into equivalent one-way and two-way
deterministic and nondeterministic finite automata.
For all the costs upper and matching lower bounds have been proved, with the only exception of (a) and (b), for
which the best known lower and upper bounds are, respectively, exponential and double exponential.}
\label{fig:diagram}
\end{figure}

\section{Witness Languages and Variants of $1$-Limited Automata}
\label{sec:variant}

As mentioned in the introduction, \las can be very succinct. In fact, for some languages the size gap to \owdfa is double exponential.
We already observed that this gap is related to nondeterminism. Indeed, if nondeterministic choices
are not possible, the gap reduces to a single exponential (see \cref{fig:diagram}).
However,
we want to understand better on the one hand how much we can restrict the model,
still keeping this double exponential gap and,
on the other hand,
if there is a restriction that,
still allowing some kind of nondeterminism,
reduces the gap to a single exponential.

In our investigations,
the following language,
which is defined with respect to an integer parameter~$n>0$,
will be useful:
\[
	K_n=\{x_1\cdots x_k\cdot x\mid k>0,\; x_1,\ldots,x_k,x\in\{a,b\}^n,\; \exists j\in\{1,\ldots,k\},\; x_j=x\}\,.
\]
We point out that each string in the language is a list of blocks of length~$n$. We ask the membership of the last block to the list of previous ones.

\begin{theorem}
\label{th:Kn}
	The language~$K_n$ is accepted by a \la with~$\bigoof{n}$ states that,
	in each accepting computation,
	replaces the content only of one cell.
\end{theorem}
\begin{proof}
	A \la~\mA[M] can scan the tape from left to right, marking a nondeterministically chosen tape cell.
	In this scan, \mA[M] can also verify that the input length is a multiple of~$n$. Furthermore, the marking can be done
	in the last cell of a block of length~$n$. For this phase~$\bigoof{n}$ states are enough.
		
	Then the machine has to compare the symbols in the last block with the symbols in the chosen one,
	namely the block which ends with the marked cell.
	This can be done by moving the head back and forth from the last block to the chosen block,
	comparing the symbols in the corresponding positions in the two blocks,
	and rejecting in case of mismatch.
	Again, this can be implemented, using a counter modulo~$n$, with~$\bigoof{n}$ states.
\end{proof}

Using standard distinguishability arguments, it can be proved that to accept~$K_n$, a \owdfa requires a number of states double
exponential in~$n$ (state lower bounds for~$K_n$ are summarized in \cref{th:WitnesslowerBounds} below).

Hence, the language~$K_n$ is a witness of the double exponential gap from \las to \owdfas.
From \cref{th:Kn}, we can notice that this gap is obtained by using the capabilities of \las in a very restricted way: during each
accepting computation, only the content of one cell is modified.
This suggested us to considering the following restricted version of \las:

\begin{definition}
	\label{def:once-marking}
	A \la is said to be \emph{once marking} if
	in each computation
	there is a unique tape cell whose input symbol~$\sigma$ is replaced with its marked version~$\marked{\sigma}$,
	while all the remaining cells are never changed.
\end{definition}

In the following, for brevity, we indicate once-marking \las and once-marking \dlas as \omlas and \domlas, respectively.

We shall consider another restriction,
in which the \la marks,
in the first visit,
every cell reached by the head.

\begin{definition}
\label{def:always-marking}
	A \la is said to be \emph{always marking} if,
	each time the head visits a tape cell for the first time, 
	it replaces the input symbol~$\sigma$ in it with its marked version~$\marked{\sigma}$.
\end{definition}

In the following, for brevity, we indicate always-marking \las and always-marking \dlas as \amlas and \damlas, respectively.

We point out that \omlas and \amlas use the work alphabet~$\Gamma=\Sigma \cup \marked\Sigma\cup\{\lend,\rend\}$.
Hence, the relevant parameter for evaluating the size of these devices is their number of states,
differently than \las, in which the size of the work alphabet is not fixed.

\medskip

We present another language that will be used in the paper.
As~$K_n$,
it is defined with respect to a fixed integer~$n>0$:
\[
	J_n=\{x\cdot x_1\cdots x_k\mid k>0,\; x_1,\ldots,x_k,x\in\{a,b\}^n,\; \exists j\in\{1,\ldots,k\},\; x_j=x\}\,.
\]
Even in this case, a string is a list of blocks of length~$n$.
Here we ask the membership of the first block to the subsequent list.
Notice that~$J_n$ is the reversal of~$K_n$.

We have the following lower bounds:
\begin{theorem}
\label{th:WitnesslowerBounds}
	Let~$n>0$ be an integer.
	\begin{itemize}
	\item To accept~$J_n$, \owdfas and \ownfas need at least~$2^n$ states, while \twnfas need at least~$2^{\frac{n-1}{2}}$ states.
	\item To accept~$K_n$, \owdfas need~$2^{2^n}$ states, \ownfas need at least~$2^n$ states, and \twnfas need at least~$2^{\frac{n-1}{2}}$ states.
	\end{itemize}
\end{theorem}
\begin{proof} (sketch)
	The lower bounds for one-way machines can be proved using standard distinguishability arguments and the fooling set technique~\cite{Bir92} (see~\cite{PP14,PPS22} for similar proofs with slightly different languages).
	
	Using a standard conversion, from a~$k$-state \twnfa accepting~$K_n$ we can obtain an equivalent \owdfa with no more than~$2^{k+k^2}$ states~\cite{RS59,She59}.
	Since every \owdfa accepting~$K_n$ should have at least~$2^{2^n}$ states, we get that~$k+k^2\geq 2^n$.
	Hence~$k$ grows as an exponential in~$n$. In particular, it can be verified that~$k>2^\frac{n-1}{2}$.
	Since
	from each \twnfa accepting a language
	we can easily obtain a \twnfa with a constant amount of extra states
	accepting the reversal of such a language,
	we can conclude that the number of states of each \twnfa accepting~$J_n$ or~$K_n$ 
	must be at least exponential in~$n$.	
\end{proof}

\section{Once-Marking 1-Limited Automata}
\label{sec:OM}

During each computation,
\emph{once-marking 1-limited automata} are able to mark just one input cell.

From \cref{th:Kn},
we already know that the language~$K_n$ can be accepted by a \omla with~$\bigoof{n}$ states.
We now show that such a machine can be turned in a even more restricted form:

\begin{theorem}
\label{th:KnOnceMarking}
	The language~$K_n$ is accepted by a \omla with~$\bigoof{n}$ states
	that is sweeping and uses nondeterministic transitions only in the first traversal of the tape.
\end{theorem}
\begin{proof}
	We discuss how to modify the $\bigoof{n}$-state \omla~\mA[M] described in the proof of \cref{th:Kn}
	in order obtain a sweeping machine that uses nondeterministic transitions only in the first sweep.
	\mA[M] makes a first scan of the input, exactly as described in the proof of \cref{th:Kn}.
	In this scan the head direction is never changed.
	When the right end-marker is reached,
	\mA[M] makes~$n$ iterations,
	which in the following description will be counted from~$0$ to~$n-1$.
	
	The purpose of the iteration~$i$,
	$i=0,\ldots,n-1$,
	is to compare the $(n-i)$th symbols of the last block and of the chosen one. 
	To this aim,
	the iteration starts with the head on the right end-marker,
	and uses a counter modulo~$n$,
	initialized to~$(i+1)\bmod n$.
	The counter is decremented while moving to the left.
	In this way,
	it contains~$0$ exactly while visiting the~$(n-i)$th cell of each input block.
	Hence,
	the automaton can easily locate the~$(n-i)$th symbols of the last block and of the chosen one and check if they are equal.
	Once the left end-marker is reached, \mA[M] can cross the tape from left to right, remembering the number~$i$ of the iteration.
	Notice that~\mA[M] does not need to keep this number while moving from right to left. Indeed the value of~$i$ can be recovered from
	the value of the counter when the left end-marker is reached. 
	
	Once the iteration~$i$ is completed, if the last check was unsuccessful then~\mA[M] can stop and reject. Otherwise it can start the next
	iteration, if~$i<n-1$, or accepts.
	
	From the discussion above,
	it can be easily verified that
	\mA[M] is sweeping,
	makes nondeterministic choices only in the first sweep,
	and has~$\bigoof{n}$ many states.
\end{proof}

We now study the size relationships between \omlas and finite automata.
First, we observe that \omlas can be simulated by \ownfas and by \owdfas at the costs of an exponential and a double exponential 
increase in the number of states, respectively. These upper bounds derive from the costs of the simulations of \las by finite
automata presented in~\cite[Thm.~2]{PP14}. By considering the language~$K_n$, we can conclude that these costs cannot be reduced:

\begin{theorem}
	\label{th:boundsOnceMaking}
	Let~\mA[M] be a~$n$-state \omlas. Then~\mA[M] can be simulated by a \ownfa and by a \twnfa with a number of states exponential in~$n$,
	and by a \owdfa with a number of states double exponential in~$n$. In the worst case these costs cannot be reduced.
\end{theorem}
\begin{proof}
	The upper bounds derive from the cost of the simulations of \las by \ownfas and \owdfas given in~\cite[Thm.~2]{PP14}.
	For the lower bounds we consider the language~$K_n$.
	As proved in \cref{th:KnOnceMarking}, this language can be accepted by a \omla with~$\bigoof{n}$ states.
	Furthermore, according to \cref{th:WitnesslowerBounds}, it requires a number of state exponential in~$n$ to be accepted
	by \ownfas or \twnfas, and a number of states double exponential in~$n$ to be accepted by \owdfas.
\end{proof}

From \cref{th:boundsOnceMaking},
it follows that the ability of marking only once can give already a huge descriptional power.
Furthermore, from \cref{th:KnOnceMarking}, we can observe that this power is achievable even with a sweeping
machine that does not use nondeterminism after the first sweep.
From the size costs of the simulation of \las by finite automata (see \cref{fig:diagram}), we already know
that nondeterminism is essential to obtain this huge descriptional power.
We now prove that, without nondeterminism, the descriptional power on \omlas dramatically reduces:

\begin{theorem}
\label{th:det-OM}
	For each~$n$-state \domla there exists an equivalent \twdfa with~$\bigoof{n^3}$ states.
\end{theorem}
\begin{proof}
Let~$\mA=(Q,\Sigma,\Gamma,\delta,q_I,F)$ be a~$n$-state \domla.
We give a construction of an equivalent \twdfa~\mA*.
Before doing that,
let us introduce,
from an high-level perspective,
how the simulating machine works.

The \twdfa~\mA* operates in different modes.

In the first part of the computation,
before~\mA marks \emph{one} cell,
\mA* is in \beforeMarking mode,
in which it simulates directly
each transition of~\mA.

When~\mA* has to simulate the transition~$\delta(s,\sigma) = (\marked\sigma, d)$ used by \mA for marking a cell,
besides changing its state and moving its head according to the transition,
\mA* switches to \afterMarking mode
and stores in its finite control
the symbol~$\sigma$ that has been marked
and the state~$s$ in which~\mA was immediately before the marking.

While in \afterMarking mode,
every time a cell is visited,
\mA* has to select which transition of \mA to simulate depending on the symbol~$a$ scanned by the input head.
There are two possibilities:
if the scanned symbol is different than the symbol~\markedSymbol that has been marked,
then the transition is simulated directly.
Otherwise,
\mA* switches to \simulation mode
(described later)
to verify whether the current cell is the one that has been marked by~\mA.
If this is the case,
then \mA* simulates the transition of~\mA on the marked symbol~$\marked\sigma$,
otherwise it simulates the transition on~$\sigma$.
In both cases \mA* keeps working in \afterMarking mode,
so selecting transitions according to the strategy described above,
until there are no more moves to simulate.
Therefore~\mA* accepts if the last simulated transition corresponds to a right transition passing the right end-marker while simulating a final state of~\mA.

\bigbreak 

We now give some details on the \simulation mode,
which is the core of the simulation.
We remind the reader that~\mA* switches to this mode when,
being in \afterMarking mode,
the input head is on a cell containing the symbol~$\sigma$,
which has been saved at the end of the \beforeMarking mode.
Let us indicate by~$j$ the current position of the head,
namely the position that has to be verified.

The \twdfa~\mA* has to verify
whether~$j$ is the cell that has been marked by~\mA.
To make this check,
\mA* can verify whether the computation path of~\mA on the given input reaches,
from the initial configuration,
a configuration with state~\markingState and the head on the currently scanned cell~$j$
(we remind the reader that \markingState and \markedSymbol have been saved in the control of~\mA*
when switching from \beforeMarking to \afterMarking mode),
whose position,
however,
cannot be saved in the control.

To be sure that the machine does not ``loses track'' of the position~$j$ while performing this search,
we use the following strategy:
\begin{itemize}
	\item
		\mA* simulates a backward computation from the state~\markingState and the current position~$j$.
	\item
		If the initial configuration of~\mA is reached,
		then the cell from which the check has started is the one where the marking transition has been executed.
	\item
		At that point,
		the position~$j$ is recovered by ``rolling back'' the backward computation.
		This is done by repeating the (forward) computation of~\mA from the initial configuration until a marking transition is used.
		In fact,
		since \mA is deterministic and once marking,
		this transition is necessarily the one that,
		from the state~\markingState,
		marked~\markedSymbol.
		In other words,
		the forward computation of \mA that is simulated here is the same simulated in \beforeMarking mode. 
\end{itemize}
As we shall explain later,
even in the case the initial configuration of~\mA is not reached
(namely the verification is not successful),
our technique allows to recover the head position~$j$ from which the backward simulation started,

It is important to observe two key points for which this approach works.
The first one is that \omlas mark only one cell during their computation.
The second observation is that the simulated machine is deterministic.
Therefore,
along every accepting computation path from the initial configuration,
it occurs only once that the symbol~\markedSymbol is scanned while \mA is in state~\markingState,
which is when~\mA makes a marking transition.

To make such a verification,
and in particular the backward search,
we use a technique originally introduced by Sipser~\cite{Sip80}.
This simulation has been then refined
by Geffert, Mereghetti, and Pighizzini,
which proved that~\twdfas can be made halting with a linear increase of the number of states~\cite{GMP07}.
In the following,
we shall refer to the latter simulation as the \emph{original simulation} and use the notation and terminology contained in~\cite{GMP07},
to which we address the interested reader for missing details.

The main difference with the original simulation
is that there the simulating machine starts from the final configuration of the simulated device,
because the goal is to verify the presence of an accepting computation path.
In our case,
the machine~\mA* starts the backward simulation from the state~$s$ and the cell containing~$\sigma$ that has to be checked. 

In the following,
a~\emph{configuration} is a pair $(q,i)$,
where~$q$ is the current state
and $i$ is the position of the tape head.

Consider the graph whose nodes represent configurations and edges computation steps.
Since \mA is deterministic,
the component of the graph containing~$(s,j)$ is a tree rooted at this configuration,
with backward paths branching to all possible predecessors of $(s,j)$.
In addition,
no backward path starting from~$(s,j)$ can loop
(hence, it is of finite length),
because the marking configuration~$(s,j)$ cannot be reached by a forward path from a loop
(due to the fact that the machine is deterministic).

The simulating machine~\mA* can perform a depth-first search of this tree in order to detect whether the initial configuration~$(q_I,0)$
belongs to the predecessors of~$(s,j)$.
If this is the case,
then the machine returns to the position~$j$,
by performing a forward simulation of~\mA from~$(q_I,0)$ until when~$s$ is entered while reading the symbol~$\sigma$.
We stress that this approach works because the simulated machine is deterministic.
After that,
the simulation of \mA in \afterMarking mode is recovered by performing a move on the symbol~$\marked\sigma$. 
On the other hand,
if the whole tree has been examined without reaching~$(q_I,0)$,
then the cell in position~$j$ is not the marked one,
so the machine simulates a move of~\mA* on~$\sigma$
from the cell in position~$j$,
again switching back to \afterMarking mode.
Notice that this case occurs when there are no more predecessors of~$(s,j)$ to visit.
So,
in this case,
the machine~\mA* completes the depth-first search on the cell in position~$j$,
while looking for further nodes of the graph reachable from the configuration~$(s,j)$.
Hence,
no extra steps are required to retrieve the position~$j$.

In conclusion,
\mA* has three state components of size~\bigoof{n}:
one used in \beforeMarking and \afterMarking for the direct simulation of the transitions of~\mA,
one for storing the state~\markingState and the symbol~\markedSymbol,
and one used in \simulation mode.
So, the total number of states of~\mA* is~\bigoof{n^3}.
\end{proof}

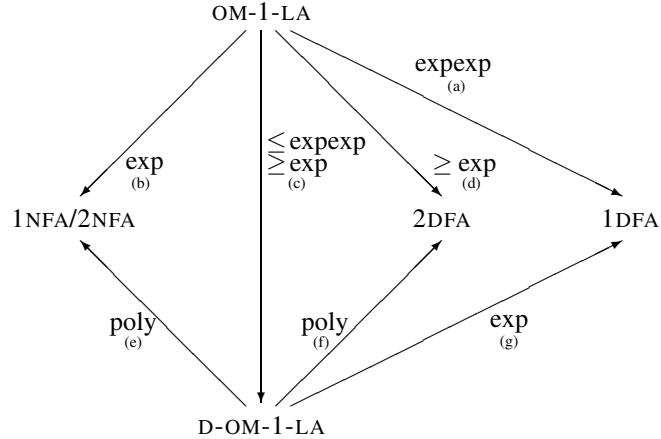
\begin{figure}[t]
	\setlength{\unitlength}{0.5cm}
\centering
	\begin{picture}(17,12)(0,0)
	
	\put(6.5,0.5){\makebox(0,0){\small\sc \domla}}
	\put(1.5,6){\makebox(0,0){\small\ownfa/\twnfa}}
	\put(11.3,6){\makebox(0,0){\small\twdfa}}
	\put(16.3,6){\makebox(0,0){\small\owdfa}}
	\put(6.5,11.5){\makebox(0,0){\small\omla}}

	\put(6,11){\vector(-1,-1){4.4}} 
	\put(3.3,7.4){\makebox(0,0){\small$\exp$}}
	\put(3.2,6.9){\makebox(0,0){\tiny (b)}}

	\put(6.4,11){\vector(0,-1){9.9}} 
	\put(7.85,8){\makebox(0,0){\small$\leq\!\text{expexp}$}}
	\put(7.35,7.4){\makebox(0,0){\small$\geq\!\exp$}}
	\put(7.35,6.9){\makebox(0,0){\tiny(c)}}

	\put(6.8,11){\vector(1,-1){4.4}} 
	\put(11.8,7.4){\makebox(0,0){\small$\geq\exp$}}
	\put(12,6.9){\makebox(0,0){\tiny(d)}}

	\put(7.2,11){\vector(2,-1){8.8}} 
	\put(11.5,10){\makebox(0,0){\small expexp}}
	\put(11.6,9.5){\makebox(0,0){\tiny(a)}}

	\put(6,1){\vector(-1,1){4.4}} 
	\put(3,3.2){\makebox(0,0){\small$\poly$}}
	\put(3,2.7){\makebox(0,0){\tiny (e)}}

	\put(6.8,1){\vector(1,1){4.4}} 
	\put(8.1,3.2){\makebox(0,0){\small$\poly$}}
	\put(8,2.7){\makebox(0,0){\tiny (f)}}

	\put(7.2,1){\vector(2,1){8.8}} 
	\put(13,3.2){\makebox(0,0){\small$\exp$}}
	\put(13,2.7){\makebox(0,0){\tiny(g)}}

	\end{picture}
\caption{Size costs of conversions involving \omlas.
	The gaps~(a) and~(b) derive from \cref{th:boundsOnceMaking}.
	For (c) and (d) the lower bound derives from the lower bound of the language~$K_n$ on \twnfas~(\cref{th:WitnesslowerBounds});
	the best known upper bound derives from~(a).
	The bounds~(e) and~(f) are from \cref{th:det-OM}.
	The upper bound for~(g) derives from the conversion from \dlas
	and the lower bound from the conversion from \twdfas.
}
\label{fig:diagramOM}
\end{figure}

In \cref{fig:diagramOM} the state costs of the conversions involving \omlas are summarized.
In particular,
we proved that the size gap from~\omlas to~\twnfas is exponential and to~\owdfas is double exponential,
while \domlas and \twdfas are polynomially related in size.

Some questions remain open, in particular about the costs of
the simulations of \omlas by \domlas and by \twdfas. At the moment, from the above mentioned results, we can derive double
exponential upper bounds and exponential lower bounds. 
The same questions are open for the simulation of \las by \dlas and by \twdfas, namely by dropping the once-marking restriction.
We point out that these questions are related to the problem of the cost of the elimination of nondeterminism from two-way finite
automata, proposed by Sakoda and Sipser in 1978~\cite{SS78}, which is still open.

\section{Always-Marking 1-Limited Automata}
\label{sec:AM}
Always-marking 1-limited automata replace,
when they visit each cell for the first time,
the input symbol with its marked version.
In this section we study the descriptional complexity of these devices.

First of all,
we prove that \amlas cannot achieve the same succinctness as~\las.
In fact,
the size gap to \owdfas reduces from double exponential for \las to single exponential.

\begin{theorem}
\label{th:upperNonDetAM}
  Each~$n$-state \amla can be simulated by a \ownfa with at most~$n\cdot 2^{n^2}$ states
  and by a complete \owdfa with at most~$(2^n-1)\cdot 2^{n^2}+1$ states.
\end{theorem}

\begin{proof}
  Let $\mA[M]=(Q,\Sigma,\Gamma,\delta,q_0,F)$ be a given~$n$-state \amla.
  We adapt the argument used in~\cite{PP14} to convert \las into
  \ownfas and \owdfas,
  which is derived from the technique to convert \twdfas into equivalent~\owdfas,
  presented in~\cite{She59},
  and based on \emph{transitions tables}.
  
  Roughly,
  transition tables represent the possible behaviors of~\mA[M] on frozen tape segments.
  More precisely, given~$z\in\Gamma^*$ ,
  the \emph{transition table} associated with~$z$ is the binary
  relation~$\tau_z\subseteq Q\times Q$,
  consisting of all pairs~$(p,q)$ such that~\mA[M] has a computation path that
  starts in the state~$p$ on the rightmost symbol of the tape segment containing $\lend z$,
  ends entering the state~$q$ by leaving the same tape segment to the right side, i.e.,
  by moving from the rightmost cell of the segment to the right, and
  does not visit any cell outside the segment.
  
  First, we can apply the conversion presented in~\cite{PP14} from 
  \las to \ownfas, in order to obtain from~\mA[M] an equivalent \ownfa~$A$, whose
  computations simulate the computations of~\mA[M] by keeping in the finite state control two components:
  \begin{itemize}
  \item The transition table associated with the part of the tape at the left of the head.
  	This part has been already visited and, hence, it is frozen.
  \item The state in which the simulated computation of~\mA[M] reaches the current tape position for the first time.
  \end{itemize}
  For details we address the reader to~\cite[Thm.~2]{PP14}.
  Since the number of transition tables is at most~$2^{n^2}$, the number of states in the resulting \ownfa~$A$ is bounded by~$n\cdot 2^{n^2}$.
  
  Applying the subset construction, this automaton can be converted into an equivalent deterministic
  one, with an exponential increase of the number of states, so obtaining a double exponential 
  number of states in~$n$.
  In the general case, this increasing cannot be reduced.
  This is due to the fact that different computations of~$A$, after reading the same input, could
  keep in the control different transitions tables, depending on the fact that~\mA[M] can replace
  the same input by different strings.
  
  However, under the restriction we are considering, along different computations, each input string~$x$
  is always replaced by the same string~$\marked{x}$,  which is obtained by marking every symbol of~$x$.
  Hence, at each step of the simulation, the transition table stored by~$A$ depends only on the  
input prefix already inspected. The only part
  that can change is the state of the simulated computation of~\mA[M] after reading~$x$.
  
  This allows to obtain from~$A$ a \owdfa~$A'$, equivalent to \mA[M] that, after reading a string~$x$, has in its finite state control
  the transition table associated with~$\marked{x}$, and the \emph{set} of states that the computations of~\mA[M] can reach after
  reading~$x$. In other words, the automaton~$A'$ is obtained from~$A$ by keeping the first component of the
  control, which is deterministic, and making a subset construction for the second one.
  
  By summarizing, the possible values of the first component are~$2^{n^2}$, while the values of the second one
  are~$2^n$, namely the possible subsets of the state set of~\mA[M]. This gives a~$2^n\cdot 2^{n^2}$ upper bound.
  We can slightly reduce this number, by observing that when the second component contains the empty set,
  i.e., each computation of~\mA[M] (or equivalently of~$A$) stops before entering it, then the input is rejected, regardless
  the first component. Hence, we can replace all the pairs having the empty set as a second component with
  a unique sink state, so reducing the upper bound to~$(2^n-1)\cdot 2^{n^2}+1$
\end{proof}

The asymptotical optimality of the upper bounds in \cref{th:upperNonDetAM} derives from the optimality of the conversions
from \twnfas to \ownfas and to \twdfas~\cite{RS59,She59,Kap05}.

We now show that \amlas can be more succinct than \twnfas, even in the deterministic case. 
In particular we prove the following:

\begin{theorem}
\label{th:DAM->2DFA}
	The language~$J_n$ is accepted by a \damla with~$\bigoof{n}$ states, while it cannot be accepted 
	by any \twnfa with less than~$2^{\frac{n-1}2}$ states.
\end{theorem}
\begin{proof}
	The lower bound for \twnfas has been given in \cref{th:WitnesslowerBounds}.
	The possibility of marking the already-visited cells allows to reduce this cost, 
	even without making use of the nondeterminism, as we now describe.
	An always marking \dla~\mA[M] can firstly visit and mark the first~$n$ tape cells.
	Then, it starts to inspect the next block of length~$n$.
	When the head reaches for the first time a cell, \mA[M] remembers the
	scanned symbol~$\sigma$ in it and moves the head back to
	the left end-marker and then to the corresponding cell in the first block (this can be implemented
	with a counter modulo~$n$). If the symbol in this cell is not~$\sigma$ then~\mA[M] has to skip
	the remaining symbols in the block under inspection and inspect the next block, if any.
	This can be done moving the head to the left end-marker and then, starting
	to count modulo~$n$, moving to the right until finding the first symbol of the next block.
	This symbol can be located using the value of the counter and the fact that it has not been marked yet.
	Otherwise, if the symbol in the cell coincides with~$\sigma$ and the block is not completely inspected
	(see below), 
	\mA[M] moves the head to the right to
	search the next symbol of the block under inspection, namely the first unmarked symbol.
	
	When locating a symbol,
	\mA[M] can also check and remember if it is in position~$n$.
	This is useful to detect whether a block has been completely scanned,
	which also means that the block has been \emph{successfully scanned},
	otherwise the machine would have already rejected.
	Hence,
	in this case,
	\mA[M] can move the head to the right to finally reach the accepting configuration.
	However,
	according to the definition of~$J_n$,
	before doing that,
	\mA[M] needs to verify that the input has length multiple of~$n$.
	All these steps can be implemented with a fixed number of variables and a counter modulo~$n$.
	This allows to conclude that~\mA[M] can be implemented with~$\bigoof{n}$ states.
\end{proof}

In \cref{th:DAM->2DFA} we proved an exponential gap from \damlas to \twnfas
and hence also to one-way finite automata. This allows to conclude that the following upper bounds,
that are immediate consequences of the corresponding upper bounds for \dlas~\cite[Thm.~2]{PP14}, cannot be significantly reduced:

\begin{theorem}\label{th:upperDetAM}
  Each~$n$-state \damla can be simulated by a \owdfa and by a \ownfa with no more than~$n\cdot(n+1)^n$ states.
\end{theorem}

From the  discussion above and \cref{th:upperDetAM},
we have the same state gap from \damlas and from \dlas to one-way automata. 

\begin{figure}[t]
	\setlength{\unitlength}{0.5cm}
\centering
	\begin{picture}(17,12)(0,0)
	
	\put(6.5,0.5){\makebox(0,0){\small\sc \damla}}
	\put(1.5,6){\makebox(0,0){\small\ownfa/\twnfa}}
	\put(11.3,6){\makebox(0,0){\small\twdfa}}
	\put(16.3,6){\makebox(0,0){\small\owdfa}}
	\put(6.5,11.5){\makebox(0,0){\small\amla}}

	\put(6,11){\vector(-1,-1){4.4}} 
	\put(3.3,7.4){\makebox(0,0){\small$\exp$}}
	\put(6.4,11){\vector(0,-1){9.9}} 
	\put(7.5,7.4){\makebox(0,0){\small$\leq\exp$}}
	\put(7.7,6.9){\makebox(0,0){\tiny(a)}}

	\put(6.8,11){\vector(1,-1){4.4}} 
	\put(11.3,7.4){\makebox(0,0){\small$\exp$}}
	\put(7.2,11){\vector(2,-1){8.8}} 
	\put(11.5,10){\makebox(0,0){\small $\exp$}}
	\put(6,1){\vector(-1,1){4.4}} 
	\put(3,3.2){\makebox(0,0){\small$\exp$}}
	\put(6.8,1){\vector(1,1){4.4}} 
	\put(8.3,3.2){\makebox(0,0){\small$\exp$}}
	\put(7.2,1){\vector(2,1){8.8}} 
	\put(13,3.2){\makebox(0,0){\small$\exp$}}

	\end{picture}
\caption{State costs of conversions involving \amlas.
All the exponential upper bounds derive from \cref{th:upperNonDetAM,th:upperDetAM},
while the lower bounds derive from \cref{th:DAM->2DFA}.
For~(a) we do not know if in the worst case an exponential size is also necessary.
}
\label{fig:diagramAM}
\end{figure}
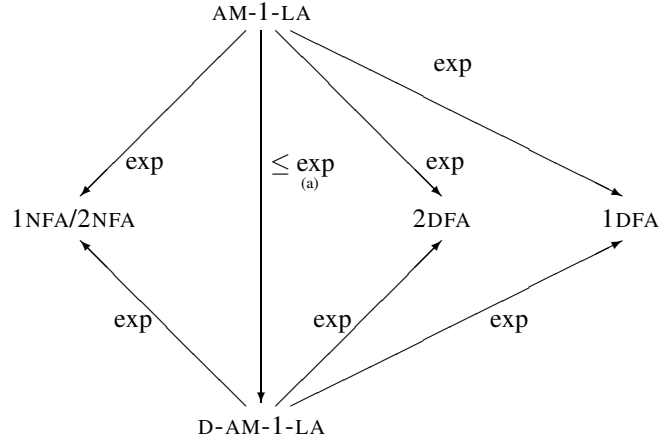

\medskip

The state costs of the conversions involving \amlas are summarized in \cref{fig:diagramAM}.

Even in the case of~\amlas,
as well as in the cases of \las and of \omlas,
we do not know how much the elimination of the nondeterminism costs.
Here,
we have an exponential upper bound for the conversion of \amlas into \damlas but,
at the moment,
we do not have a matching lower bound.
Considering the conversion of \amlas into \twdfas, unlikely the analogous conversions from \las and \omlas,
here we have matching exponential upper and lower bounds.
As already mentioned at the end of \cref{sec:OM}, these questions are related to the open
question of Sakoda and Sipser.

\section{Conclusion}
We study the costs of the simulations of \omlas and
\amlas by finite automata. \cref{fig:diagramOM,fig:diagramAM} give a summary of the
results we obtained. They can be compared with the costs of the simulations concerning \las,
in \cref{fig:diagram}.

We observed that \amlas cannot reach the same succinctness as \las and \omlas\linebreak
(see \cref{th:boundsOnceMaking,th:upperNonDetAM}). 
In particular, in \cref{th:KnOnceMarking} we have shown that the language~$K_n$
can be accepted by a \omla with~$\bigoof{n}$ states. Hence, it requires an exponential
number of states on \amlas due to the fact that a double exponential
number of states on \owdfas is necessary (see \cref{th:WitnesslowerBounds}).
It is not difficult to describe a \twnfa accepting~$K_n$ with an exponential number of states.
We point out that such a machine is also a~\amla.
Hence,
by summarizing,
the language~$K_n$ is accepted by a \omla with~$\bigoof{n}$ states,
by an~\amla with a number of states exponential in~$n$,
and by a \owdfa with a number of states double exponential in~$n$.
All these costs cannot be reduced.

Since in the nondeterministic case the gaps from \omlas to finite automata are the same as from \las,
a natural question is to ask if \omlas are always as succinct as \las.
Intuitively the answer to this question is negative.
For instance we do not see how to recognize the language whose strings are concatenations of blocks of length~$n$,
in which two blocks are equal,
with a~\omla with~$\bigoof{n}$ states,
while it is not hard to accept it using a \la with such a number of states. 
We leave the study of this question for a future work.

Another candidate for studying this question is the unary language~$(a^{2^n})^*$.
We proved that this language can be accepted by a \dla with~$\bigoof{n}$ states and a work
alphabet of cardinality~$\bigoof{n}$, and by a \dla with ~$\bigoof{n^3}$ states and 
work alphabet of size not dependent on~$n$~\cite{PP19,PP23}. As pointed out in~\cite{PP19},
each \twnfa accepting it requires at least~$2^n$ states. Hence, by \cref{th:det-OM}
even each \domla accepting it requires an exponential number of states.
We do not see how to reduce this number even by allowing the use of nondeterminism on 
\omlas or on \amlas.

More in general, the comparisons between the sizes of these restricted versions of \las deserve further
investigation, even in the unary case where the cost of several simulations are still unknown~\cite{PP19}. 
In a recent paper,
we investigated \emph{forgetting \las},
another restriction of \las in which there is a unique symbol~$X$ that is used to replace input symbols.
Therefore,
during the first visit to a cell,
its original content is always replaced by~$X$~\cite{PP23b}.

Finally, we would like to mention once again the problem of the cost of removing nondeterminism from \las,
\omlas, and \amlas (see \cref{sec:OM,sec:AM}),
which is connected to the main question of the 
cost of the elimination of
nondeterminism from two-way finite automata,
raised longtime ago by Sakoda and Sipser and still open~\cite{SS78}
(for a survey, see~\cite{Pi13}).

\bibliographystyle{eptcs}
\bibliography{biblio}

\end{document}